\documentclass[journal]{IEEEtran}
\usepackage{stmaryrd}
\usepackage{amsfonts}
\usepackage{graphicx,times,amsmath}
\usepackage{graphics,color}
\usepackage{graphicx}
\usepackage{latexsym}
\usepackage{amsmath}
\usepackage{amsfonts}
\usepackage{amssymb}
\usepackage{url}
\usepackage{subfigure}
\usepackage{times}
\usepackage{color}
\newtheorem{theorem}{Theorem}
\newtheorem{corollary}{Corollary}
\newtheorem{definition}{Definition}
\newtheorem{lemma}{Lemma}
\newtheorem{remark}{Remark}

\allowdisplaybreaks

\begin{document}

\title{Fixed-time cluster synchronization for complex networks via pinning control}

\author{Xiwei~Liu,~{\it{Member,~IEEE,}}~
        and~Tianping~Chen,~{\it{Senior~Member,~IEEE}}
\thanks{This work is jointly supported by the National Natural Sciences Foundation of China (Nos. 61203149, 61273211 and 61233016), the National Basic Research Program of China (973 Program) under
Grant No. 2010CB328101, ``Chen Guang'' project supported by Shanghai
Municipal Education Commission and Shanghai Education Development
Foundation under Grant No. 11CG22, and the Fundamental Research Funds
for the Central Universities.}
\thanks{Xiwei Liu is with Department of Computer Science and
Technology, Tongji University, and with the Key Laboratory of Embedded System and Service Computing,
Ministry of Education, Shanghai 200092, China. E-mail:
xwliu@tongji.edu.cn; xwliu.sh@gmail.com}
\thanks{Corresponding Author Tianping Chen is with the School of Mathematical/Computer
Sciences, Fudan University, 200433, Shanghai, P.R. China. E-mail:
tchen@fudan.edu.cn}
}

\maketitle

\begin{abstract}
\boldmath {\bf In this paper, the fixed-time cluster synchronization problem for complex networks via pinning control is discussed. Fixed-time synchronization has been a hot topic in recent years, which means that the network can achieve synchronization in finite-time and the settling time is bounded by a constant for any initial values. To realize the fixed-time cluster synchronization, a simple distributed protocol by pinning control technique is designed, whose validity is rigorously proved, and some sufficient criteria for fixed-time cluster synchronization are also obtained. Especially, when the cluster number is one, the cluster synchronization becomes the complete synchronization problem; when the intrinsic dynamics for each node is missed, the fixed-time cluster synchronization becomes the fixed-time cluster (or complete) consensus problem; when the network has only one node, the coupling term between nodes will disappear, and the synchronization problem becomes the simplest master-slave case, which also includes the stability problem for nonlinear systems like neural networks. All these cases are also discussed. Finally, numerical simulations are presented to demonstrate the correctness of obtained theoretical results.}
\end{abstract}


\section{Introduction}
Synchronization of complex networks \cite{PC98}-\cite{LC15} has been a hot topic in recent decades, and it is usually described by the following continuous-time ordinary differential equation:
\begin{align*}
\dot{x}_i(t)=f(x_i(t))+g_i(x_1(t),\cdots,x_N(t)), ~~i=1,\cdots,N
\end{align*}
where $x_i\in R^n$ denotes the state of agent $i$; the first term $f(x_i(t))$ denotes the intrinsic dynamics of agent $i$, when $f=0$, the synchronization model becomes the consensus model; and the second term $g_i(x_1(t),\cdots,x_N(t))$ means the diffusive coupling from agent $i$'s neighbours. Although each agent just needs to get the local information of its neighbours, under the above algorithm, the whole network can display a collective behavior----synchronization, i.e.,
\begin{align}\label{syn}
\lim\limits_{t\rightarrow\infty}\|x_i(t)-x_j(t)\|=0; ~~i,j=1,2,\cdots,N
\end{align}
where $\|\cdot\|$ denotes some norm. It is a fundamental research topic in decentralized control, and has broad applications in cooperative control of unmanned air vehicles (UAVs), formation control of mobile robots, etc. The most popular model in the synchronization literature is the linear coupling protocol
\begin{align*}
\dot{x}_i(t)=f(x_i(t))+c\sum_{j=1}^Na_{ij}(x_j(t)-x_i(t)), ~~i=1,\cdots,N
\end{align*}
and the nonlinear coupling protocol
\begin{align*}
\dot{x}_i(t)=f(x_i(t))+c\sum_{j=1}^Na_{ij}\phi(x_j(t),x_i(t)), ~~ i=1,\cdots,N
\end{align*}

\emph{(Infinite-time vs Finite-time)} Under these coupling protocols, the sufficient criteria for asymptotic synchronization or exponential synchronization can be obtained. Synchronization rate is an important
performance indicator for protocol. However, for the above two types of synchronization speed, the disadvantage is that the completely same of each node can never occur in finite-time, i.e., there does not exist a constant $T$ called the \emph{settling time}, which may depend on the initial values, such that for any $i,j=1,2,\cdots,N$
\begin{align}\label{st}
\lim\limits_{t\rightarrow T}\|x_i(t)-x_j(t)\|=0;~\mathrm{and}~ x_i(t)=x_j(t), ~~\forall t\ge T.
\end{align}
In many applications as robotics, a standard problem in system theory is to develop controllers which drive a system to a given position as fast as possible \cite{H86}. Moreover, chaos synchronization plays an important role in the literature, if the synchronization does not realize in a finite-time, for example, exponential synchronization is realized, then the coupling protocol or the external controllers should be always added on the network, since if they are cancelled, from the property of chaotic oscillator, a small error may cause a large difference between nodes. Furthermore, finite-time synchronization can lead to better system performances in the disturbance rejection and robustness against uncertainty. Based on above reasons, an investigation of finite-time synchronization under new coupling protocols is important both in theoretical analysis and real applications.

Next, we review some progress in the finite-time literature. Finite-time synchronization (or consensus) heavily relies on the development of finite-time stabilization theory. Finite-time stabilization problems have been studied mostly in the contexts of optimality, controllability, and deadbeat control for several decades.
These control laws are usually time-varying, discontinuous, or even depending directly on the initial conditions of considered systems, for example, in \cite{C06}, the author considered the following controllers
\begin{align*}
\dot{x}=-\frac{\mathrm{grad}(f)(x)}{\|\mathrm{grad}(f)(x)\|_2}, ~~\mathrm{and}~~\dot{x}=-\mathrm{sgn}(\mathrm{grad}(f)(x)),
\end{align*}
by using the theory of Filippov, the author proved the finite-time stability and applied it on the network consensus problem. Recently, finite-time stability and finite-time stabilization via continuous time-invariant feedback have been studied. For example, \cite{BB98} studied the finite-time stability of a homogeneous systems and designed the bounded continuous finite-time stabilizing feedback laws for the double integrators; \cite{Hong02b} studied global finite-time control of robot systems through state feedback and output feedback control; \cite{BB00} discussed the finite-time stability of continuous systems, investigated its sensitivity to perturbations and rigorously set up a general framework for finite-time stabilization, many papers' works were based on this excellent result; \cite{SX08}-\cite{SH09} proposed two new sufficient conditions for local finite-time stability and designed an observer for a class of homogeneous systems with Lipschitz nonlinearity.

As for the finite-time synchronization/consensus literature, in \cite{XWCG09}, the author proposed a new class of finite-time consensus protocol:
\begin{align*}
\dot{x}_i(t)=\beta\mathrm{sig}\bigg(\sum_{j\in N_i}W_{ij}(x_j-x_i)\bigg)^{\alpha}+\gamma\sum_{j\in N_i}W_{ij}(x_j-x_i),
\end{align*}
where $0<\alpha<1, \beta>0, \gamma\ge 0$, they proved that if the network had a spanning tree, then the above protocol can realize the finite-time consensus. \cite{WX10} investigated the finite-time consensus under the protocol:
\begin{align*}
\dot{x}_i(t)=\sum_{j\in N_i}a_{ij}\mathrm{sig}(x_j-x_i)^{\alpha_{ij}(t)}, ~~~0<\alpha_{ij}(t)<1,
\end{align*}
they proved that if the sum of time intervals, in which the interaction topology was connected, was sufficiently large, the above protocol can realize finite-time consensus for both bidirected and unidirected networks. \cite{JW11} considered the finite-time weighted average consensus of time-varying topology with respect to the monotonic function $g$ under the protocol:
\begin{align*}
\dot{x}_i(t)=\frac{1}{{\omega_idg}/{d{x}_i}}
\sum_{j=1}^na_{ij}(t)\mathrm{sig}(h(x_j)-h(x_i))^{\alpha}, 0<\alpha<1,
\end{align*}
when $\omega_i=1$, then the above protocol was the one in \cite{JW09}. \cite{WH10} considered the finite-time $\chi$-consensus models,
\begin{align*}
\dot{x}_i(t)=\frac{1}{|\frac{\partial{\chi}}{\partial{x}_i}|}
\sum_{j=1}^na_{ij}(t)\psi(\mathrm{sig}(x_j-x_i)^{\alpha}), ~~~0\le \alpha<1,
\end{align*}
and
\begin{align*}
\dot{x}_i(t)=\frac{1}{|\frac{\partial{\chi}}{\partial{x}_i}|}
\sum_{j=1}^na_{ij}(t)[\psi(x_j^{\alpha})-\psi(x_i^{\alpha})], ~~~0\le \alpha<1,
\end{align*}
using the homogeneity property of $\psi$ and the Lyapunov function, the authors finally proved the finite-time consensus for undirected (the first model) and directed (the second model) networks. With the similar analysis, \cite{SCGDL12} discussed the finite-time consensus of the leader-following multi-agent systems with jointly-reachable leader and switching jointly-reachable leader for the first-order model
\begin{align*}
\dot{x}_i(t)=&\sum_{v_j\in N_i(t)}a_{ij}(t)\phi(\mathrm{sig}(x_j(t)-x_i(t))^{\alpha})\\
&-b_i(t)\phi_2(\mathrm{sig}(x_i(t)-x_0)^{\alpha}),~~~0\le \alpha<1,
\end{align*}
and for the second-order model
\begin{align*}
\dot{x}_i(t)=&v_i(t),\\
\dot{v}_i(t)=&\sum_{v_j\in N_i(t)}a_{ij}(t)[\varphi_1(\mathrm{sig}(x_j(t)-x_i(t))^{\alpha_1})\\
&~~~~~~~~~~~~~~+\varphi_2(\mathrm{sig}(v_j(t)-v_i(t))^{\alpha_2})]\\
&-b_i(t)[\varphi_3(\mathrm{sig}(x_j(t)-x_0(t))^{\alpha_1})\\
&~~~~~~+\varphi_4(\mathrm{sig}(v_j(t)-v_0(t))^{\alpha_2}))],
\end{align*}
where $\dot{x}_0(t)=v_0(t), \dot{v}_0(t)=0$ and $0<\alpha_1<1, \alpha_2=\frac{2\alpha_1}{1+\alpha_1}$. Moreover, in \cite{ZDW15}, by using the tools from homogeneous theory, the authors investigated finite-time consensus with second-order integrators based on both relative position and relative velocity measurements as
\begin{align*}
\dot{q}_i(t)=&p_i(t),\\
\dot{p}_i(t)=&-k_1\tanh\bigg[\mathrm{sig}\bigg(\sum_{j=1}^Na_{ij}
(q_i-q_j)\bigg)^{\alpha_1}\bigg]\\
&-k_2\tanh\bigg[\mathrm{sig}
\bigg(\sum_{j=1}^Na_{ij}(p_i-p_j)\bigg)^{\alpha_2}\bigg],
\end{align*}
where $0<\alpha_1<1, \alpha_2=\frac{2\alpha_1}{1+\alpha_1}$,
and they also considered the case with only the relative position measurements obtained. In \cite{LHYC14}, the optimal finite-time stabilization problem was considered, and they designed a new switching protocol covering both continuous control ($0<\alpha<1$) and discontinuous one ($\alpha=0$), which can shorten the stabilization time; with the same aim, the authors in \cite{LLYC15} also considered the optimal consensus time for multi-agents systems and proposed a new switching protocol. Moreover, in \cite{WLS2014}, the distributed finite-time containment control algorithm for double-integrator multi-agent systems with multiple dynamic or stationary leaders was proposed and its validity was also rigorously proved.

\emph{(Finite-time vs Fixed-time)} Although many finite-time results are obtained in the above short review, the settling time heavily depends on the initial conditions, which limits the practical applications, since the knowledge of initial conditions is not available in advance. Therefore, to overcome this drawback, in \cite{P12}, a new concept called the \emph{fixed-time stability} is proposed, if it is globally finite-time stable and the settling time function is bounded for any initial values. The technique is adding extra terms on the previous finite-time model, which will be stated more carefully in the next section. Following this streamline of dealing with fixed-time stability, many new approaches and results are obtained, see \cite{PPS12}-\cite{Z15} and references therein. For example, \cite{ZT14a} and \cite{ZT14} proposed two new protocols to realize the fixed-time stabilization, and the upper bounds for the settling time were also estimated; moreover, they also applied these new protocols on fixed-time consensus problem for multi-agent systems. \cite{Z15} proposed a fixed-time terminal sliding-mode control protocol for a class of second-order nonlinear systems in the presence of uncertainties and perturbations. In \cite{MJD15}, the authors considered the finite-time consensus for multi-agent systems with cooperative and antagonistic interactions, and they proved that the states of all agents can reach agreement in a finite-time regarding consensus values that were the same in modulus but may not be the same in sign.

Although many papers have considered the finite-time and fixed-time consensus problem, the discussions about synchronization and consensus have a great difference because of the existence of intrinsic dynamics. There are few works to discuss the finite-time and fixed-time synchronization. For example, \cite{CL09} investigated the finite-time complete synchronization, while \cite{ZS15} studied the fixed-time complete synchronization. On the other hand, cluster synchronization is a more practical phenomenon than the complete synchronization, which is significant in biological sciences, communication engineering, etc. The cluster synchronization means that nodes in the same cluster can achieve complete synchronization, while nodes in different clusters have different dynamical behaviors. Of course, when all the nodes lie in the same cluster, then the cluster synchronization becomes the complete synchronization. There were some paper considering the cluster synchronization exponentially like \cite{WZC09,LC2011}, \cite{MLZ2006}-\cite{LLC2010}, and \cite{CFZW14} also considered the finite-time cluster synchronization of Markovian switching complex networks with stochastic perturbations by adding the linear negative feedback controllers.

{\bf{To our best knowledge, there are no papers considering the \emph{fixed-time} cluster synchronization via pinning control, which will be the subject of this paper.}}

This paper is organized as follows. In Section \ref{pre}, some necessary definitions, lemmas, assumptions and notations are given. A simple coupling protocol to ensure the fixed-time cluster synchronization via pinning control is also proposed. In Section \ref{PR1}, we first rigorously prove the effectiveness of proposed protocol, then the fixed-time complete synchronization is also carefully discussed. Furthermore, some numerical simulations are given in Section \ref{sim} to show the correctness of obtained theoretical results. Finally, a conclusion and some discussion about future work are presented in Section \ref{con}.

\section{Preliminaries}\label{pre}
In this section, we present some definitions, lemmas, and notations, which will be useful throughout this paper.
First, we will give some general results on the dynamical systems.

If a nonnegative function $V(t)$ satisfies
\begin{eqnarray*}
\dot{V}(t)=-\alpha\mu(V(t)),
\end{eqnarray*}
where $\alpha>0$, functions $\mu(V(t))>0$, $V(t)>0$; $\mu(0)=0$.

Because $\dot{V}(t)> 0$, therefore, $V(t)$ is decreasing.

Define the function $V_{1}(t)$ as follows:
\begin{align*}
V_{1}(s)=-\alpha^{-1}\int_{V(0)}^{s}\mu^{-1}(V)dV,
\end{align*}
then,
\begin{align}
t= V_{1}(V(t)),
\end{align}
and
\begin{align}
V_{1}^{-}(t)= V(t),
\end{align}
where $V_{1}^{-}$ is the inverse function of $V$.

If $\mu(s)=s$, then
\begin{align*}
t=V_{1}(V(t))=-\alpha^{-1}\int_{V(0)}^{V(t)}V^{-1}dV=-\alpha^{-1}\log\frac{V(t)}{V(0)},
\end{align*}
and
$$V(t)=V(0)e^{-\alpha t}.$$

If
$\mu(s)=s^{p}$, $p\ne 1$, then
$$V_{1}(s)=\frac{-1}{\alpha(1-p)}(s^{1-p}-V^{1-p}(0)),$$
which implies
$$V(t)=[\alpha(p-1)t+V^{1-p}(0)]^{\frac{1}{1-p}}.$$

In case $p<1$, then $V(t)=0$, if
$$t\ge\frac{V^{1-p}(0)}{\alpha(1-p)}.$$

On the other hand, in case $p>1$,
$$V(t)=\frac{1}{[\alpha(p-1)t
+V^{1-p}(0)]^{\frac{1}{p-1}}}.$$

As direct consequences, we show some useful lemmas on the finite-time and fixed-time stability.

\begin{lemma} (See \cite{BB00})
If a nonnegative function $V(t)$ satisfies
\begin{align}
\dot{V}(t)\le
-\alpha V^{p}(t),~0<p<1
\end{align}
where $\alpha>0$. Then, $V(t)\equiv 0$, if
\begin{align}
t\ge \frac{V^{1-p}(0)}{\alpha(1-p)}.
\end{align}
\end{lemma}

Many papers investigating the finite-time stability or consensus are based on this result, see \cite{XWCG09}-\cite{LLYC15}. Moreover, the other forms like $\dot{V}(t)\le -\alpha V^p(t)\pm kV(t)$ are also proposed and proved to realize finite-time stability in \cite{SX08} and \cite{SH09}, here we omit them, interested readers can refer to these works.

However, the settling time depends on the initial value $V(0)$, which in many cases cannot be obtained. Therefore, in our recent paper, a general criteria for the fixed-time stability is proposed, which can be stated as follows.

\begin{lemma} (See \cite{LLC16})\label{t2}
If a nonnegative function $V(t)$ satisfies
\begin{eqnarray}
\dot{V}(t)\le\left\{\begin{array}{lll}
-\alpha V^{p}(t),~0<p<1           & ;&\mathrm{if}~ 0<V<1\\
-\beta V^{q}(t),~q>1      & ;&\mathrm{if}~ V\geq 1
\end{array}\right.
\end{eqnarray}
where $\alpha>0, \beta>0$. Then, $V(t)\equiv 0$, if
\begin{align}
t\ge \frac{1}{\alpha(1-p)}+\frac{1}{\beta(q-1)}.
\end{align}
\end{lemma}

\begin{remark}
As a direct result, if a nonnegative function $V(t)$ satisfies
\begin{align}
\dot{V}(t)\le -\alpha V^{p}(t)-\beta V^{q}(t),
\end{align}
where $\alpha>0, \beta>0, 0<p<1, q>1$. Then, $V(t)\equiv 0$, if
\begin{align}
t\ge \frac{1}{\alpha(1-p)}+\frac{1}{\beta(q-1)}.
\end{align}
\end{remark}

\begin{remark}
From the above two lemmas, one can conclude that for the finite-time stability, only one term like $-V(t)^{p}, 0<p<1$ can realize this aim; however, to realize fixed-time stability, except for the term for finite-time stability, one should add an extra term $-V(t)^q, q>1$, whose role can be regarded as pulling the system into the region with norm less than $1$ in a fixed-time.
\end{remark}

\begin{remark}
Some other criteria to ensure the fixed-time stability, like
\begin{align}\label{f1}
\dot{V}(t)\le -(\alpha V^{p}(t)+\beta V^{q}(t))^{k},
\end{align}
where $\alpha>0,~\beta>0.~p>0,~q>0,~k>0,~0<pk<1,~qk>1$, which is presented in \cite{P12}; or
\begin{align}\label{f2}
\dot{V}(t)\le -\alpha V^{1-\frac{1}{2\mu}}(t)-\beta V^{1+\frac{1}{2\mu}}(t),~~\mu>1,
\end{align}
which is proposed in \cite{PPS12}; or
\begin{align}\label{f3}
\dot{V}(t)\le -\alpha V^{\frac{m}{n}}(t)-\beta V^{\frac{p}{q}}(t),
\end{align}
where $m,n,p,q$ are all positive odd integers satisfying $m>n$ and $p<q$, which is proposed in \cite{ZT14}, can all be regarded the special case of Lemma \ref{t2}. Of course, utilizing the concrete form of (\ref{f2}) and (\ref{f3}), the authors can obtain a more exact estimation of the settling time.
\end{remark}

Since the dealing with different norms happens in the fixed-time and finite-time literature, we first introduce a powerful lemma.
\begin{lemma} (See \cite{HL52})\label{norm}
For any vector $z\in R^n$ and $0<r<l$, the following norm equivalence property holds:
\begin{align}\label{pro1}
\bigg(\sum_{i=1}^n |z_i|^l\bigg)^{1/l}\le \bigg(\sum_{i=1}^n|z_i|^r\bigg)^{1/r},
\end{align}
and
\begin{align}\label{pro2}
\bigg(\frac{1}{n}\sum_{i=1}^n |z_i|^l\bigg)^{1/l}\ge \bigg(\frac{1}{n}\sum_{i=1}^n|z_i|^r\bigg)^{1/r}.
\end{align}
\end{lemma}

\begin{remark}
The first inequality (\ref{pro1}) is called the Jensen inequality, whose proof can be found in P. 4, \cite{HL52}, and this inequality is also commonly used in the finite-time literature. On the other hand, the second inequality (\ref{pro2}) is very useful to deal with the newly added term like $-V(t)^q, q>1$, whose proof can be found in P. 26, \cite{HL52}. The above two inequalities can be combined in the norm form as follows:
\begin{align}\label{pro}
\|z\|_l\le \|z\|_{r}\le n^{\frac{1}{r}-\frac{1}{l}}\|z\|_l,
\end{align}
where $\|z\|_r=(\sum_{i=1}^n|z_i|^r)^{1/r}$ and $\|z\|_l=(\sum_{i=1}^n|z_i|^l)^{1/l}$.
\end{remark}

Next, we will introduce some definitions and lemmas about the cluster synchronization in complex networks.

For a complex network of $N$ nodes, suppose its graph is  $\mathcal{G}=\{\mathcal{V},\mathcal{E}\}$,
where $\mathcal{V}$ represents the vertex set numbered by $\left\{1,\cdots,N\right\}$,
and $\mathcal{E}$ denotes the edge set with $e(i,j)\in \mathcal{E}$ if and only if there is an edge from vertex
$j$ to $i$.

\begin{definition}(See \cite{LC2011})\label{mydef1}
Matrix $A=(a_{ij})\in R^{N\times N}$ is said to belong to class $A1$, denoted as $A\in A1$, if
\begin{enumerate}
\item $a_{ij}\ge 0,i\neq j,\ a_{ii}=-\sum_{j=1,j\neq i}^Na_{ij},\ i=1,\cdots,N$
\item $A$ is irreducible.
\end{enumerate}
Furthermore, if $A\in A1$ and $a_{ij}=a_{ji},\ i\neq j$, then we say matrix $A$ belong to class $A2$, denoted as $A\in A2$.

Obviously, if the Laplacian matrix for the graph $G$ is $A$, then $A\in A1$ means that the graph is a strongly connected directed graph; and $A\in A2$ means that the graph is a strongly connected undirected graph.
\end{definition}


\begin{lemma}\label{pin}(See \cite{CLL07})
Suppose $A\in A1$ and $\epsilon<0$, then, there exists a positive definite diagonal matrix $\Theta=\mathrm{diag}(\theta_1,\cdots,\theta_N)$, such that the matrix $A_{\epsilon}=A+\mathrm{diag}\{\epsilon,0,\cdots,0\}$ is Lyapunov stable, i.e.,
\begin{align*}
\Theta A_{\epsilon}+A^T_{\epsilon}\Theta<0.
\end{align*}
Especially, when $A\in A2$ and $\epsilon<0$, we can choose $\Phi=I$, where $I$ is the identity matrix with appropriate dimensions, such that matrix $A_{\epsilon}$ is negative definite, i.e., its eigenvalues are all negative and can be sorted as
\begin{align}
0>\lambda_1(A_{\epsilon})\ge\lambda_2(A_{\epsilon})\ge\lambda_3(A_{\epsilon})
\ge\cdots\ge\lambda_N(A_{\epsilon}).
\end{align}
\end{lemma}

\begin{lemma}\label{use} (See \cite{CZ2007})
Suppose $A=(a_{ij})_{N\times N}\in A2$, then for any two vectors $X=(x_1,\cdots,x_N)^T$ and $Y=(y_1,\cdots,y_N)^T$, we have
\begin{align}
X^TAY=-\sum_{j>i}a_{ij}(x_j-x_i)(y_j-y_i).
\end{align}
\end{lemma}

\begin{definition}(See \cite{LC2011})\label{rowsum}
Matrix $A=(a_{ij})\in R^{N_1\times N_2}$ is said to belong to class $A3$, denoted as $A\in A3$, if its each row-sum is zero, i.e.,
$\sum_{j=1}^{N_2}a_{ij}=0, i=1,\dots,N_1$.
\end{definition}

\begin{remark}
In fact, we can also assume that each row-sum is a non-zero constant, but for convenience, we assume the row-sum is zero, which means that the interactions between nodes can be cooperative (when the element is positive) and competitive (when the element is negative).
\end{remark}

In order to investigate the cluster synchronization, we assume the set of nodes in the network can be divided into $m$ clusters, i.e.,
$\{1,\cdots,N\}=\mathcal{C}_1\cup\mathcal{C}_2\cup \cdots\cup\mathcal{C}_m$, where
\begin{align}\label{clust}
\mathcal{C}_1=\{1,\cdots,r_1\}, ~\mathcal{C}_2=\{r_1+1,\cdots,r_2\},~\cdots,\nonumber\\
\mathcal{C}_k=\{r_{k-1}+1,\cdots,r_{k}\},~\cdots, ~\mathcal{C}_m=\{r_{m-1}+1,\cdots,N\}.
\end{align}

Now, using the above definitions of matrices, we can define a new type of coupling matrix $A$ for the cluster synchronization analysis.

\begin{definition}\label{matrix}
Suppose $A\in R^{N\times N}$ is \emph{symmetric}, the indexes $\{1,\cdots,N\}$ can be divided into $m$ clusters as defined in (\ref{clust}), and the following form holds
\begin{align}\label{form}
A=\left(
\begin{array}{cccc}
A_{11}&A_{12}&\cdots&A_{1m}\\
A_{21}&A_{22}&\cdots&A_{2m}\\
\vdots&\vdots&\ddots&\vdots\\
A_{m1}&A_{m2}&\cdots&A_{mm}
\end{array}
\right),
\end{align}
where $A_{ij}\in R^{(r_i-r_{i-1})\times(r_j-r_{j-1})}, r_0=0$, $A_{ii}\in A2$ and $A_{ij}\in A3$, $i,j\in \{1,\cdots,m\}$. Then the matrix $A$ is said to belong to class $A4$, denoted as $A\in A4$.
\end{definition}

\begin{remark}
In fact, the coupling matrix for cluster synchronization is also defined in \cite{LC2011}, but in that paper, the matrix $A$ can be asymmetric. The reason for the requirement in the above definition is for the convenience of later analysis. Moreover, in \cite{LLC2010}, the authors also investigate another type of cluster synchronization, i.e., $A\in A1$ and $A_{ij}\in A3$, obviously, they are different, since in this type of coupling matrix, the interactions between nodes are only cooperative ones.
\end{remark}

Cluster synchronization means that: each vertex in the same cluster has the same individual node dynamic, while nodes in different clusters are nonidentical, which can guarantee that the trajectories are apparently distinguishing when cluster synchronization is reached.
In \cite{LLC2010}, the authors use different intrinsic dynamics to guarantee the final cluster synchronization, while \cite{WZC09}-\cite{LC2011} use the pinning control technique \cite{CLL07} to realize the final cluster synchronization.
Now, for the complex network with clusters defined in (\ref{clust}), we give the following definition of fixed-time cluster synchronization.
\begin{definition}
For the network with (\ref{clust}), the \emph{fixed-time cluster synchronization} is said to be realized, if there exists a time $T$ independent on the initial values, such that for any initial values, node $i$ can converge to the target trajectory $s_k(t)$, which belongs to the $k$-th cluster $\mathcal{C}_k$, $k=1,2,\cdots,m$, i.e.,
\begin{align}
\lim\limits_{t\rightarrow T}\|x_i(t)-s_k(t)\|=0, ~\mathrm{and}~ x_i(t)=s_k(t), ~\forall t\ge T,
\end{align}
where target trajectories in different clusters are different with each other, i.e., $s_{k_1}(t)\ne s_{k_2}(t), k_1,k_2\in \{1,\cdots,m\}$.
\end{definition}

Now, we are in the position to propose the coupling protocols for the fixed-time cluster synchronization via pinning control. Without loss of generality, we assume the controllers are added on just the first node of each cluster. Therefore, we list the protocol for cluster $\mathcal{C}_k, k=1,\cdots,m$ as follows:
\begin{align}\label{M1}
&\dot{x}_{r_{k-1}+1}(t)=f(x_{r_{k-1}+1}(t))\nonumber\\
+&\alpha \sum_{j\in\mathcal{C}_k} a_{r_{k-1}+1,j}\mathrm{sig}^p\bigg(x_j(t)-x_{r_{k-1}+1}(t)\bigg)\nonumber\\
+&\sum_{k^{\prime}\ne k}\mathrm{sig}^p\bigg(\sum_{j\in\mathcal{C}_{k^{\prime}}} a_{r_{k-1}+1,j}(x_j(t)-x_{r_{k-1}+1}(t))\bigg)\nonumber\\
+&\beta\sum_{j\in\mathcal{C}_k}b_{r_{k-1}+1,j}\mathrm{sig}^q
\bigg(x_j(t)-x_{r_{k-1}+1}(t)\bigg)\nonumber\\
+&\sum_{k^{\prime}\ne k}\mathrm{sig}^q\bigg(\sum_{j\in\mathcal{C}_{k^{\prime}}} b_{r_{k-1}+1,j}(x_j(t)-x_{r_{k-1}+1}(t))\bigg)\nonumber\\
-&\epsilon_1\mathrm{sig}^p(x_{r_{k-1}+1}(t)-s_k(t))\nonumber\\
-&\epsilon_2\mathrm{sig}^q(x_{r_{k-1}+1}(t)-s_k(t)),
\end{align}
and for $i=r_{k-1}+2,\cdots,r_k$,
\begin{align}\label{M2}
\dot{x}_{i}(t)=&f(x_{i}(t))\nonumber\\
+&\alpha \sum_{j\in\mathcal{C}_k} a_{ij}\mathrm{sig}^p\bigg(x_j(t)-x_{i}(t)\bigg)\nonumber\\
+&\sum_{k^{\prime}\ne k}\mathrm{sig}^p\bigg(\sum_{j\in\mathcal{C}_{k^{\prime}}} a_{ij}(x_j(t)-x_{i}(t))\bigg)\nonumber\\
+&\beta\sum_{j\in\mathcal{C}_k}b_{ij}\mathrm{sig}^q
\bigg(x_j(t)-x_{i}(t)\bigg)\nonumber\\
+&\sum_{k^{\prime}\ne k}\mathrm{sig}^q\bigg(\sum_{j\in\mathcal{C}_{k^{\prime}}} b_{ij}(x_j(t)-x_{i}(t))\bigg),
\end{align}
where $x_i=(x_i^1,\cdots,x_i^n)^T\in R^n$, scalars $\alpha>0, \beta>0, \epsilon_1>0, \epsilon_2>0, 0<p<1, q>1$, and the target trajectory $s_k(t)$ in cluster $\mathcal{C}_k$ is governed by
\begin{align}\label{tar}
\dot{s}_k(t)=f(s_k(t)),~~~ k=1,2,\cdots,m,
\end{align}
where targets in different clusters are different, i.e., $s_{k_1}(t)\ne s_{k_2}(t), k_1\ne k_2\in \{1,\cdots,m\}$.

Function $f: R^n\rightarrow R^n$ is used to describe the intrinsic behavior of each isolated node, which satisfies the following QUAD condition:
\begin{align}\label{QU}
(x-y)^T(f(x)-f(y))\le\delta (x-y)^T(x-y),~~ \delta>0.
\end{align}

For any vector $x=(x^1,\cdots,x^n)^T$, the function $\mathrm{sig}^p(x): R^n\rightarrow R^n$ is defined as:
\begin{align}
\mathrm{sig}^p(x)=(\mathrm{sign}(x^1)|x^1|^p,\cdots,\mathrm{sign}(x^n)|x^n|^p)^T.
\end{align}
Especially, when $p=1$, $\mathrm{sig}(x)=x$.


\begin{remark}
In fact, we have made many simple assumptions. For example, as for the QUAD condition (\ref{QU}), one can assume that there exists a diagonal matrix $\Delta=\mathrm{diag}(\delta_1,\cdots,\delta_n)$ and $\delta_i>0, i=1,\cdots,n$, such that $(x-y)^T(f(x)-f(y))\le (x-y)^T\Delta(x-y)$. Moreover, as for the coupling term, one can also use a function $\Psi(x)$ instead of $x$, where $\Psi(x)=(\psi_1(x^1),\cdots,\psi_n(x^n))^T$, therefore in this case, the nonlinear function includes the linear one with $\Psi(x)=\Gamma x$, where $\Gamma$ is a positive diagonal matrix. Furthermore, the parameters $\alpha, \beta, \epsilon_1, \epsilon_2$ defined in (\ref{M1}) and (\ref{M2}) can also be different values or vectors which are functions of the cluster $k$ as: $\alpha(k), \beta(k), \epsilon_1(k), \epsilon_2(k)$.

In all, to simply the process and emphasis on the key role of terms $\mathrm{sig}^p(\cdot), 0<p<1$ and $\mathrm{sig}^q(\cdot), q>1$ for the fixed-time synchronization, we adopt the assumptions as given, interested readers can generalize these cases themselves.
\end{remark}

\begin{lemma}\label{young} (Young's inequality)
Suppose $a, b, u, v$ are all positive scalars, and $\frac{1}{u}+\frac{1}{v}=1, u>1, v>1$, then
\begin{align}
ab\le \frac{a^u}{u}+\frac{b^v}{v}
\end{align}
where ``='' holds if and only if $a^u=b^v$.
\end{lemma}

\section{Main results}\label{PR1}
\subsection{Cluster synchronization with $m>1$}
At first, we define the synchronization error in cluster $\mathcal{C}_k, k=1,2,\cdots,m$ as:
\begin{align}
e_i(t)=x_i(t)-s_k(t),
\end{align}
where $i=r_{k-1}+1,\cdots,r_{k}$ and $s_k(t)$ is the target trajectory defined by (\ref{tar}). For $i\in \mathcal{C}_k$, define
\begin{align}
\tilde{f}(e_i)=f(x_i(t))-f(s_k(t)).
\end{align}

Therefore, the dynamics of the synchronization error in cluster $\mathcal{C}_k, k=1,\cdots,m$ is given as follows:
\begin{align}\label{E1}
&\dot{e}_{r_{k-1}+1}(t)=\tilde{f}(e_{r_{k-1}+1}(t))\nonumber\\
+&\alpha \sum_{j\in\mathcal{C}_k} a_{r_{k-1}+1,j}\mathrm{sig}^p\bigg(e_j(t)-e_{r_{k-1}+1}(t)\bigg)\nonumber\\
+&\sum_{k^{\prime}\ne k}\mathrm{sig}^p\bigg(\sum_{j\in\mathcal{C}_{k^{\prime}}} a_{r_{k-1}+1,j}e_j(t)\bigg)\nonumber\\
+&\beta\sum_{j\in\mathcal{C}_k}b_{r_{k-1}+1,j}\mathrm{sig}^q
\bigg(e_j(t)-e_{r_{k-1}+1}(t)\bigg)\nonumber\\
+&\sum_{k^{\prime}\ne k}\mathrm{sig}^q\bigg(\sum_{j\in\mathcal{C}_{k^{\prime}}} b_{r_{k-1}+1,j}e_j(t)\bigg)\nonumber\\
-&\epsilon_1\mathrm{sig}^p(e_{r_{k-1}+1}(t))
-\epsilon_2\mathrm{sig}^q(e_{r_{k-1}+1}(t)),
\end{align}
and for $i=r_{k-1}+2,\cdots,r_k$,
\begin{align}\label{E2}
&\dot{e}_{i}(t)=\tilde{f}(e_{i}(t))\nonumber\\
+&\alpha \sum_{j\in\mathcal{C}_k} a_{ij}\mathrm{sig}^p\bigg(e_j(t)-e_{i}(t)\bigg)+\sum_{k^{\prime}\ne k}\mathrm{sig}^p\bigg(\sum_{j\in\mathcal{C}_{k^{\prime}}} a_{ij}e_j(t)\bigg)\nonumber\\
+&\beta\sum_{j\in\mathcal{C}_k}b_{ij}\mathrm{sig}^q
\bigg(e_j(t)-e_{i}(t)\bigg)+\sum_{k^{\prime}\ne k}\mathrm{sig}^q\bigg(\sum_{j\in\mathcal{C}_{k^{\prime}}} b_{ij}e_j(t)\bigg),
\end{align}

Before giving the main result, we define some new matrices. Suppose the coupling matrices $A$ and $B$ in (\ref{E1}) and (\ref{E2}) satisfies that $A\in A4$ and $B\in A4$ in Definition \ref{matrix} with the same cluster partition defined in (\ref{clust}). For $k=1,\cdots,m$, we define $\overline{A}_{kk}=(\overline{a}_{ij})_{(r_k-r_{k-1})\times (r_{k}-r_{k-1})}$ as: for $i,j=1,\cdots,r_k-r_{k-1}$,
\begin{align}
\overline{a}_{ij}=\left\{
\begin{array}{ll}
a_{r_{k-1}+i,r_{k-1}+j}^{\frac{2}{1+p}}; &i\ne j\\
-\sum_{j^{\prime}=1,j^{\prime}\ne i}^{r_k-r_{k-1}}a_{r_{k-1}+i,r_{k-1}+j^{\prime}}^{\frac{2}{1+p}}; &i=j
\end{array}\right.
\end{align}
We also define $\overline{B}_{kk}=(\overline{b}_{ij})_{(r_k-r_{k-1})\times (r_{k}-r_{k-1})}$ as: for $i,j=1,\cdots,r_k-r_{k-1}$,
\begin{align}
\overline{b}_{ij}=\left\{
\begin{array}{ll}
b_{r_{k-1}+i,r_{k-1}+j}^{\frac{2}{1+q}}; &i\ne j\\
-\sum_{j^{\prime}=1,j^{\prime}\ne i}^{r_k-r_{k-1}}b_{r_{k-1}+i,r_{k-1}+j^{\prime}}^{\frac{2}{1+q}}; &i=j
\end{array}\right.
\end{align}
Obviously, $\overline{A}_{kk}\in A2$ and $\overline{B}_{kk}\in A2, k=1,\cdots,m$. Moreover, define
\begin{align}
\hat{A}_{kk}=-2\overline{A}_{kk}+\mathrm{diag}\{(2\epsilon_1\alpha^{-1})^{\frac{2}{1+p}},0,\cdots,0\},\label{pg1}\\
\hat{B}_{kk}=-2\overline{B}_{kk}+\mathrm{diag}\{(2\epsilon_2\beta^{-1})^{\frac{2}{1+q}},0,\cdots,0\}.\label{pg2}
\end{align}
Since $\overline{A}_{kk}\in A2$ and $\overline{B}_{kk}\in A2$, according to Lemma \ref{pin}, matrices $\hat{A}_{kk}$ and $\hat{B}_{kk}$ are positive definite, and we use $\lambda_{\min}(\hat{A}_{kk})$ and $\lambda_{\min}(\hat{B}_{kk})$ to denote the smallest eigenvalue of $\hat{A}_{kk}$ and $\hat{B}_{kk}$, respectively. Denote positive scalars
\begin{align}
\rho_1=\min_{k=1,\cdots,m}\lambda_{\min}(\hat{A}_{kk}),
\rho_2=\min_{k=1,\cdots,m}\lambda_{\min}(\hat{B}_{kk});\label{rho}\\
\overline{N}=n\bigg(\sum_{k=1}^m\frac{(r_k-r_{k-1})(r_k-r_{k-1}-1)}{2}+m\bigg);\nonumber\\
\overline{\alpha}=\alpha2^{\frac{p-1}{2}}\rho_1^{\frac{1+p}{2}}, ~~~\overline{\beta}=\beta \overline{N}^{\frac{1-q}{2}}2^{\frac{q-1}{2}}\rho_2^{\frac{1+q}{2}};\label{g1}\\
\overline{a}=\max_{i\in \mathcal{C}_{k}, j\in\mathcal{C}_{k^{\prime}},k\ne k^{\prime}}|a_{ij}|,~~~ \overline{b}=\max_{i\in \mathcal{C}_{k}, j\in\mathcal{C}_{k^{\prime}},k\ne k^{\prime}}|b_{ij}|;\nonumber\\
\overline{r}=\max_{k=1,\cdots,m}[N-(r_k-r_{k-1})];\nonumber\\
\gamma_1=\overline{a}^p\overline{r}(Nn)^{\frac{1-p}{2}}2^{\frac{1+p}{2}},
~~~\gamma_2=\overline{b}^q\overline{r}2^{\frac{1+q}{2}}.\label{g2}
\end{align}

Now, with the above notations, we will prove the following main theorem.
\begin{theorem}\label{main}
For coupled systems (\ref{M1}) and (\ref{M2}), if $A\in A4$, $\overline{\alpha}-\gamma_1-2\delta>0$ and $\overline{\beta}-\gamma_2-2\delta>0$, then the fixed-time cluster synchronization can be achieved with the settling time defined as:
\begin{align}\label{settling}
T_{max}=\frac{2}{(\overline{\alpha}-\gamma_1-2\delta)(1-p)}+\frac{2}{(\overline{\beta}-\gamma_2-2\delta)(q-1)},
\end{align}
where $\delta$ is as defined in (\ref{QU}), and positive parameters $\overline{\alpha}, \overline{\beta}, \gamma_1, \gamma_2$ are defined in (\ref{g1}) and (\ref{g2}).
\end{theorem}

\begin{proof}
Define the Lyapunov function as
\begin{align}\label{lya}
V(t)=\frac{1}{2}\sum_{k=1}^m\sum_{i=r_{k-1}+1}^{r_k}e_i(t)^Te_i(t)
=\frac{1}{2}\sum_{k=1}^mV_k(t),
\end{align}
where $V_k(t)=\sum_{i=r_{k-1}+1}^{r_k}e_i(t)^Te_i(t)=E_k(t)^TE_k(t)$, and $E_k(t)=(e_{r_{k-1}+1}(t)^T, \cdots, e_{r_k}(t)^T)^T, k=1,\cdots,m$.

Differentiating it, we have
\begin{align*}
&\dot{V}(t)\\
=&\sum_{k=1}^m\sum_{i=r_{k-1}+1}^{r_k}e_i(t)^T\dot{e}_i(t)
=\sum_{k=1}^m\sum_{i=r_{k-1}+1}^{r_k}e_i(t)^T\bigg[\tilde{f}(e_{i}(t))\\
+&\alpha \sum_{j\in\mathcal{C}_k} a_{ij}\mathrm{sig}^p\bigg(e_j(t)-e_{i}(t)\bigg)
+\sum_{k^{\prime}\ne k}\mathrm{sig}^p\bigg(\sum_{j\in\mathcal{C}_{k^{\prime}}} a_{ij}e_j\bigg)\\
+&\beta\sum_{j\in\mathcal{C}_k}b_{ij}\mathrm{sig}^q
\bigg(e_j(t)-e_{i}(t)\bigg)+\sum_{k^{\prime}\ne k}\mathrm{sig}^q\bigg(\sum_{j\in\mathcal{C}_{k^{\prime}}} b_{ij}e_j\bigg)\bigg]\\
+&\sum_{k=1}^me_{r_{k-1}+1}^T\bigg[-\epsilon_1\mathrm{sig}^p(e_{r_{k-1}+1}(t))
-\epsilon_2\mathrm{sig}^q(e_{r_{k-1}+1}(t))\bigg]\\
=&\sum_{k=1}^m\sum_{i=r_{k-1}+1}^{r_k}e_i(t)^T\tilde{f}(e_{i}(t))\\
+&\alpha\sum_{k=1}^m\sum_{i=r_{k-1}+1}^{r_k}e_i(t)^T\sum_{j\in\mathcal{C}_k} a_{ij}\mathrm{sig}^p\bigg(e_j(t)-e_{i}(t)\bigg)\\
+&\sum_{k=1}^m\sum_{i=r_{k-1}+1}^{r_k}e_i(t)^T\sum_{k^{\prime}\ne k}\mathrm{sig}^p\bigg(\sum_{j\in\mathcal{C}_{k^{\prime}}} a_{ij}e_j(t)\bigg)\\
+&\beta\sum_{k=1}^m\sum_{i=r_{k-1}+1}^{r_k}e_i(t)^T\sum_{j\in\mathcal{C}_k}b_{ij}\mathrm{sig}^q
\bigg(e_j(t)-e_{i}(t)\bigg)\\
+&\sum_{k=1}^m\sum_{i=r_{k-1}+1}^{r_k}e_i(t)^T\sum_{k^{\prime}\ne k}\mathrm{sig}^q\bigg(\sum_{j\in\mathcal{C}_{k^{\prime}}} b_{ij}e_j(t)\bigg)\\
-&\epsilon_1\sum_{k=1}^me_{r_{k-1}+1}(t)^T\mathrm{sig}^p(e_{r_{k-1}+1}(t))\\
-&\epsilon_2\sum_{k=1}^me_{r_{k-1}+1}(t)^T\mathrm{sig}^q(e_{r_{k-1}+1}(t))\\
=&\tilde{V}_1(t)+\tilde{V}_2(t)+\tilde{V}_3(t)+\tilde{V}_4(t)+\tilde{V}_5(t)+\tilde{V}_6(t)+\tilde{V}_7(t).
\end{align*}

According to (\ref{QU}), one can get
\begin{align}\label{v1}
\tilde{V}_1(t)\le \delta\sum_{k=1}^m\sum_{i=r_{k-1}+1}^{r_k}e_i(t)^Te_i(t)=2\delta V(t).
\end{align}

Using (\ref{pro1}) in Lemma \ref{norm}, let $r=(1+p)/2\in (0,1)$, $l=1$ and $r<l$, and combining with Lemma \ref{use}, one can get that
\begin{align}\label{v2v6}
&\tilde{V}_2(t)+\tilde{V}_6(t)\nonumber\\
=&\frac{\alpha}{2}\sum_{k=1}^m\sum_{i,j\in\mathcal{C}_k}a_{ij}(e_i(t)-e_j(t))^T
\mathrm{sig}^{p}\bigg(e_j(t)-e_i(t)\bigg)\nonumber\\
&-\epsilon_1\sum_{k=1}^me_{r_{k-1}+1}(t)^T\mathrm{sig}^p(e_{r_{k-1}+1}(t))\nonumber\\
=&-\frac{\alpha}{2}\sum_{l=1}^n
\sum_{k=1}^m\sum_{i,j\in\mathcal{C}_k}a_{ij}|e_j^l(t)-e_i^l(t)|^{1+p}\nonumber\\
&-\epsilon_1\sum_{l=1}^n\sum_{k=1}^m|e_{r_{k-1}+1}^l(t)|^{1+p}\nonumber\\
\le&-\frac{\alpha}{2}\bigg[\sum_{l=1}^n\sum_{k=1}^m\bigg(\sum_{i,j\in\mathcal{C}_k}a_{ij}^{\frac{2}{1+p}}|e_j^l(t)-e_i^l(t)|^2\nonumber\\
&~~~~~~~~~~~~~~+(2\epsilon_1\alpha^{-1})^{\frac{2}{1+p}}e_{r_{k-1}+1}^l(t)^2\bigg)\bigg]^{\frac{1+p}{2}}\nonumber\\
=&-\frac{\alpha}{2}\bigg[\sum_{k=1}^m\bigg(E_k(t)^T[(-2\overline{A}_{kk})\otimes I]E_k(t)\nonumber\\
&+E_k(t)^T[\mathrm{diag}\{(2\epsilon_1\alpha^{-1})^{\frac{2}{1+p}},0,\cdots,0\}\otimes I]E_k(t)\bigg)\bigg]^{\frac{1+p}{2}}\nonumber\\
=&-\frac{\alpha}{2}\bigg[\sum_{k=1}^mE_k(t)^T[\hat{A}_{kk}\otimes I]E_k(t)\bigg]^{\frac{1+p}{2}}\nonumber\\
\le&-\frac{\alpha}{2}\bigg[\sum_{k=1}^m\lambda_{\min}(\hat{A}_{kk})E_k(t)^TE_k(t)\bigg]^{\frac{1+p}{2}}\nonumber\\
\le&-\frac{\alpha}{2}\bigg[\rho_1\sum_{k=1}^mE_k(t)^TE_k(t)\bigg]^{\frac{1+p}{2}}\nonumber\\
=&-\frac{\alpha}{2}\bigg[2\rho_1V(t)\bigg]^{\frac{1+p}{2}}
=-\overline{\alpha}V(t)^{\frac{1+p}{2}}.
\end{align}

Similarly, using (\ref{pro2}) in Lemma \ref{norm}, let $l=(1+q)/2>1=r$, and combining with Lemma \ref{use}, one can get that
\begin{align}\label{v4v7}
&\tilde{V}_4(t)+\tilde{V}_7(t)\nonumber\\
=&-\frac{\beta}{2}\sum_{l=1}^n
\sum_{k=1}^m\sum_{i,j\in\mathcal{C}_k}b_{ij}|e_j^l(t)-e_i^l(t)|^{1+q}\nonumber\\
&-\epsilon_2\sum_{l=1}^n\sum_{k=1}^m|e_{r_{k-1}+1}^l(t)|^{1+q}\nonumber\\
\le&-\frac{\beta}{2}\overline{N}^{\frac{1-q}{2}}\bigg[\sum_{l=1}^n\sum_{k=1}^m\bigg(\sum_{i,j\in\mathcal{C}_k}b_{ij}^{\frac{2}{1+q}}|e_j^l(t)-e_i^l(t)|^2\nonumber\\
&~~~~~~~~~~~~~~+(2\epsilon_2\beta^{-1})^{\frac{2}{1+q}}e_{r_{k-1}+1}^l(t)^2\bigg)\bigg]^{\frac{1+q}{2}}\nonumber\\
=&-\frac{\beta}{2}\overline{N}^{\frac{1-q}{2}}\bigg[\sum_{k=1}^mE_k(t)^T[\hat{B}_{kk}\otimes I]E_k(t)\bigg]^{\frac{1+q}{2}}\nonumber\\
\le&-\frac{\beta}{2}\overline{N}^{\frac{1-q}{2}}\bigg[2\rho_2V(t)\bigg]^{\frac{1+q}{2}}
=-\overline{\beta}V(t)^{\frac{1+q}{2}}.
\end{align}

According to the Young's inequality in Lemma \ref{young}, we have
\begin{align}\label{v3}
&\tilde{V}_3(t)\nonumber\\
=&\sum_{k=1}^m\sum_{k^{\prime}\ne k}\sum_{i=r_{k-1}+1}^{r_k}e_i(t)^T\mathrm{sig}^p\bigg(\sum_{j\in\mathcal{C}_{k^{\prime}}} a_{ij}e_j(t)\bigg)\nonumber\\
\le&\sum_{l=1}^n\sum_{k=1}^m\sum_{k^{\prime}\ne k}\sum_{i=r_{k-1}+1}^{r_k}|e_i^l(t)||\sum_{j\in\mathcal{C}_{k^{\prime}}} a_{ij}e_j^l(t)|^p\nonumber\\
\le&\overline{a}^p\sum_{l=1}^n\sum_{k=1}^m\sum_{k^{\prime}\ne k}\sum_{i\in\mathcal{C}_k}\sum_{j\in \mathcal{C}_{k^{\prime}}}\bigg(\frac{|e_i^l(t)|^{1+p}}{1+p}
+\frac{p|e_j^l(t)|^{1+p}}{1+p}\bigg)\nonumber\\
=&\overline{a}^p\sum_{l=1}^n\sum_{k=1}^m\sum_{i\in\mathcal{C}_k}|e_i^l(t)|^{1+p}[N-(r_k-r_{k-1})]\nonumber\\
\le&\overline{a}^p\overline{r}\sum_{k=1}^m\sum_{i\in\mathcal{C}_k}\sum_{l=1}^n\bigg(|e_i^l(t)|^2\bigg)^{\frac{1+p}{2}}\nonumber\\
\le&\overline{a}^p\overline{r}(Nn)^{\frac{1-p}{2}}\bigg[\sum_{k=1}^m\sum_{i\in\mathcal{C}_k}\sum_{l=1}^n|e_i^l(t)|^2\bigg]^{\frac{1+p}{2}}
=\gamma_1V(t)^{\frac{1+p}{2}}.
\end{align}

Similarly, 
\begin{align}\label{v5}
&\tilde{V}_5(t)\nonumber\\
=&\sum_{k=1}^m\sum_{k^{\prime}\ne k}\sum_{i=r_{k-1}+1}^{r_k}e_i(t)^T\mathrm{sig}^q\bigg(\sum_{j\in\mathcal{C}_{k^{\prime}}} b_{ij}e_j(t)\bigg)\nonumber\\
\le&\sum_{l=1}^n\sum_{k=1}^m\sum_{k^{\prime}\ne k}\sum_{i=r_{k-1}+1}^{r_k}|e_i^l(t)||\sum_{j\in\mathcal{C}_{k^{\prime}}} b_{ij}e_j^l(t)|^q\nonumber\\
\le&\overline{b}^q\sum_{l=1}^n\sum_{k=1}^m\sum_{k^{\prime}\ne k}\sum_{i\in\mathcal{C}_k}\sum_{j\in \mathcal{C}_{k^{\prime}}}\bigg(\frac{|e_i^l(t)|^{1+q}}{1+q}
+\frac{q|e_j^l(t)|^{1+q}}{1+q}\bigg)\nonumber\\
\le&\overline{b}^q\overline{r}\sum_{k=1}^m\sum_{i\in\mathcal{C}_k}\sum_{l=1}^n\bigg(|e_i^l(t)|^2\bigg)^{\frac{1+q}{2}}\nonumber\\
\le&\overline{b}^q\overline{r}\bigg[\sum_{k=1}^m\sum_{i\in\mathcal{C}_k}\sum_{l=1}^n|e_i^l(t)|^2\bigg]^{\frac{1+q}{2}}
=\gamma_2V(t)^{\frac{1+q}{2}}.
\end{align}

Therefore, from (\ref{v1})-(\ref{v5}), we have
\begin{align}
\dot{V}(t)
&\le 2\delta V(t)-(\overline{\alpha}-\gamma_1)V(t)^{\frac{1+p}{2}}-(\overline{\beta}-\gamma_2) V(t)^{\frac{1+q}{2}}\nonumber\\
&\le \left\{\begin{array}{ll}
-(\overline{\alpha}-\gamma_1-2\delta)V(t)^{\frac{1+p}{2}};&~~V(t)<1\\
-(\overline{\beta}-\gamma_2-2\delta) V(t)^{\frac{1+q}{2}};&~~V(t)\ge 1
\end{array}\right.
\end{align}
According to Lemma \ref{t2}, one can get that the fixed-time synchronization is finally realized, and the settling time can be also obtained as (\ref{settling}). The proof is completed.
\end{proof}

\begin{remark} (Cluster consensus)
Obviously, when $f(\cdot)=0$, then the above fixed-time cluster synchronization problem becomes the \emph{fixed-time cluster consensus problem}. In this case, only if $\overline{\alpha}-\gamma_1>0$ and $\overline{\beta}-\gamma_2>0$, then the fixed-time cluster consensus can be achieved with the settling time defined as:
\begin{align}\label{settling2}
T_{max}=\frac{2}{(\overline{\alpha}-\gamma_1)(1-p)}+\frac{2}{(\overline{\beta}-\gamma_2)(q-1)}.
\end{align}
\end{remark}

\begin{remark}
Let us analysis the role of parameters for fixed-time cluster synchronization. For the network model (\ref{E1}) and (\ref{E2}), since the network topology is constant with time, therefore, parameters $\gamma_1$ and $\gamma_2$ in (\ref{g2}) and $\delta$ in (\ref{QU}) are constants. If we let $\epsilon_1=\alpha\omega_1$ and $\epsilon_2=\alpha\omega_2$, then matrices $\hat{A}_{kk}$ and $\hat{B}_{kk}$ in (\ref{pg1}) and (\ref{pg2}) are constant matrices independent of the parameters $\alpha$ and $\beta$, so $\rho_1$ and $\rho_2$ in (\ref{rho}) are constants; therefore, according to condition (\ref{settling}), one can get that the larger of parameters $\alpha$ and $\beta$, the network can realize the fixed-time cluster synchronization more quickly.
\end{remark}

\begin{remark}
Expect for the protocol (\ref{M1}) and (\ref{M2}), another two possible ways to realize fixed-time cluster synchronization are replacing the terms
\begin{align*}
\sum_{k^{\prime}\ne k}\mathrm{sig}^p\bigg(\sum_{j\in\mathcal{C}_{k^{\prime}}} a_{ij}(x_j(t)-x_{i}(t))\bigg)
\end{align*}
and
\begin{align*}
\sum_{k^{\prime}\ne k}\mathrm{sig}^q\bigg(\sum_{j\in\mathcal{C}_{k^{\prime}}} b_{ij}(x_j(t)-x_{i}(t))\bigg)
\end{align*}
in (\ref{M1}) and (\ref{M2}) by terms
\begin{align*}
\mathrm{sig}^p\bigg(\sum_{k^{\prime}\ne k}\sum_{j\in\mathcal{C}_{k^{\prime}}} a_{ij}(x_j(t)-x_{i}(t))\bigg)
\end{align*}
and
\begin{align*}
\mathrm{sig}^q\bigg(\sum_{k^{\prime}\ne k}\sum_{j\in\mathcal{C}_{k^{\prime}}} b_{ij}(x_j(t)-x_{i}(t))\bigg),
\end{align*}
or by linearly coupling terms as
\begin{align*}
\sum_{k^{\prime}\ne k}\sum_{j\in\mathcal{C}_{k^{\prime}}} a_{ij}(x_j(t)-x_{i}(t))
\end{align*}
and
\begin{align*}
\sum_{k^{\prime}\ne k}\sum_{j\in\mathcal{C}_{k^{\prime}}} b_{ij}(x_j(t)-x_{i}(t)).
\end{align*}
Obviously, the combinations of these protocols are also feasible. Since there are no great differences in proving their validity for fixed-time cluster synchronization (for the linear coupling, there have been many results), here we just present these protocols, interested readers can complete the proofs.
\end{remark}

\subsection{Complete synchronization with $m=1$}
When $m=1$, i.e., all the nodes lie in the same cluster, then the cluster synchronization problem becomes the complete synchronization with pinning control. In this case, the network with $N$ strongly connected nodes can be described as:
\begin{align}\label{C1}
\dot{x}_{1}(t)=&f(x_{1}(t))+\alpha \sum_{j=1}^N a_{1j}\mathrm{sig}^p\bigg(x_j(t)-x_{1}(t)\bigg)\nonumber\\
&~~~~~~~~~~+\beta\sum_{j=1}^Nb_{1j}\mathrm{sig}^q
\bigg(x_j(t)-x_{1}(t)\bigg)\nonumber\\
&-\epsilon_1\mathrm{sig}^p(x_{1}(t)-s(t))-\epsilon_2\mathrm{sig}^q(x_{1}(t)-s(t)),
\end{align}
and for $i=2,\cdots,N$,
\begin{align}\label{C2}
\dot{x}_{i}(t)=&f(x_{i}(t))+\alpha \sum_{j=1}^N a_{ij}\mathrm{sig}^p\bigg(x_j(t)-x_{i}(t)\bigg)\nonumber\\
&~~~~~~~~~~+\beta\sum_{j=1}^Nb_{ij}\mathrm{sig}^q
\bigg(x_j(t)-x_{i}(t)\bigg),
\end{align}
where $x_i\in R^n$, $\alpha>0, \beta>0$, $0<p<1$, $q>1$, and the target trajectory $s(t)$ satisfies that:
\begin{align}\label{syn-tra}
\dot{s}(t)=f(s(t)).
\end{align}

Similarly, we define $\overline{A}=(\overline{a}_{ij})_{N\times N}$ as: for $i,j=1,\cdots,N$,
\begin{align}
\overline{a}_{ij}=\left\{
\begin{array}{ll}
a_{ij}^{\frac{2}{1+p}}; &i\ne j\\
-\sum_{k=1,k\ne i}^{N}a_{ik}^{\frac{2}{1+p}}; &i=j
\end{array}\right.
\end{align}
We also define $\overline{B}=(\overline{b}_{ij})_{N\times N}$ as: for $i,j=1,\cdots,N$,
\begin{align}
\overline{b}_{ij}=\left\{
\begin{array}{ll}
b_{ij}^{\frac{2}{1+q}}; &i\ne j\\
-\sum_{k=1,k\ne i}^{N}b_{ik}^{\frac{2}{1+q}}; &i=j
\end{array}\right.
\end{align}
Moreover, define
\begin{align}
\hat{A}=-2\overline{A}+\mathrm{diag}\{(2\epsilon_1\alpha^{-1})^{\frac{2}{1+p}},0,\cdots,0\},\label{pg3}\\
\hat{B}=-2\overline{B}+\mathrm{diag}\{(2\epsilon_2\beta^{-1})^{\frac{2}{1+q}},0,\cdots,0\}.\label{pg4}
\end{align}
Since matrices $\hat{A}$ and $\hat{B}$ are positive definite, we can use $\lambda_{\min}(\hat{A})$ and $\lambda_{\min}(\hat{B})$ to denote the smallest eigenvalue of $\hat{A}$ and $\hat{B}$, respectively. Let $\overline{N}=n(\frac{N(N-1)}{2}+1)$, and
\begin{align}
\overline{\alpha}=\alpha2^{\frac{p-1}{2}}(\lambda_{\min}(\hat{A}))^{\frac{1+p}{2}}, \overline{\beta}=\beta \overline{N}^{\frac{1-q}{2}}2^{\frac{q-1}{2}}(\lambda_{\min}(\hat{B}))^{\frac{1+q}{2}}.\label{g3}
\end{align}

\begin{theorem}\label{comp}
For coupled systems (\ref{C1}) and (\ref{C2}), if $A\in A2$, $\overline{\alpha}-2\delta>0$ and $\overline{\beta}-2\delta>0$, then the fixed-time complete synchronization can be achieved with the settling time defined as:
\begin{align}\label{settling3}
T_{max}=\frac{2}{(\overline{\alpha}-2\delta)(1-p)}+\frac{2}{(\overline{\beta}-2\delta)(q-1)},
\end{align}
where $\delta$ is as defined in (\ref{QU}), and positive parameters $\overline{\alpha}, \overline{\beta}$ are defined in (\ref{g3}).
\end{theorem}

Since the proof process is the same with that in Theorem \ref{main}, here we omit it.

\begin{remark}
The network topology can be easily relaxed to be detail-balanced \cite{WX10}, i.e., there exist some scalars $\omega_i>0$, such that $\omega_ia_{ij}=\omega_ja_{ji}$.
\end{remark}

\begin{corollary} (Complete consensus) For the following complex network,
\begin{align*}
&\dot{x}_{1}(t)\nonumber\\
=&\alpha \sum_{j=1}^N a_{1j}\mathrm{sig}^p\bigg(x_j-x_{1}\bigg)+\beta\sum_{j=1}^Nb_{1j}\mathrm{sig}^q
\bigg(x_j-x_{1}\bigg)\nonumber\\
&-\epsilon_1\mathrm{sig}^p(x_{1}(t)-s)-\epsilon_2\mathrm{sig}^q(x_{1}(t)-s),
\end{align*}
and for $i=2,\cdots,N$,
\begin{align*}
\dot{x}_{i}(t)=&\alpha \sum_{j=1}^N a_{ij}\mathrm{sig}^p\bigg(x_j-x_{i}\bigg)+\beta\sum_{j=1}^Nb_{ij}\mathrm{sig}^q
\bigg(x_j-x_{i}\bigg),
\end{align*}
where $x_i\in R$, $\alpha>0, \beta>0$, $0<p<1$, $q>1$, and the target trajectory $s\in R$ can be any constant scalar, the fixed-time complete consensus can be realized, i.e., $x_i(t)=s, i=1,\cdots,N$ when
\begin{align}
t\ge T_{\max}=\frac{2}{\overline{\alpha}(1-p)}+\frac{2}{\overline{\beta}(q-1)},
\end{align}
where positive parameters $\overline{\alpha}, \overline{\beta}$ are defined in (\ref{g3}).
\end{corollary}

When the number of network node $N=1$, then the network becomes the simplest master-slave coupled systems as:
\begin{align}\label{MS}
\dot{x}(t)=&f(x(t))-\epsilon_1\mathrm{sig}^p(x(t)-s(t))-\epsilon_2\mathrm{sig}^q(x(t)-s(t)),
\end{align}
where $x(t)\in R^n$, $0<p<1$, $q>1$, and the target trajectory $s(t)$ is defined by (\ref{syn-tra}).

\begin{theorem}\label{master-slave}
For coupled systems (\ref{MS}), if $\overline{\alpha}-2\delta>0$ and $\overline{\beta}-2\delta>0$, then the fixed-time cluster synchronization can be achieved with the settling time defined as:
\begin{align}\label{settling8}
T_{max}=\frac{2}{(\overline{\alpha}-2\delta)(1-p)}+\frac{2}{(\overline{\beta}-2\delta)(q-1)},
\end{align}
where $\delta$ is as defined in (\ref{QU}), and
\begin{align}\label{ab}
\overline{\alpha}=\epsilon_12^{\frac{1+p}{2}},  ~~~~~ \overline{\beta}=\epsilon_2n^{\frac{1-q}{2}}2^{\frac{1+q}{2}}.
\end{align}
\end{theorem}

\begin{proof}
Define the Lyapunov function
\begin{align}
V(t)=\frac{1}{2}e(t)^Te(t)=\frac{1}{2}\sum_{l=1}^ne^l(t)^2,
\end{align}
where $e(t)=x(t)-s(t)=(e^1(t),\cdots,e^n(t))^T\in R^n$.

Differentiating it, we have
\begin{align*}
&\dot{V}(t)
\\=&e(t)^T[\tilde{f}(e(t))-\epsilon_1\mathrm{sig}^p(e(t))-\epsilon_2\mathrm{sig}^q(e(t))]\\
\le &2\delta V(t)-\epsilon_1|e(t)|^{1+p}-\epsilon_2|e(t)|^{1+q}\\
=&2\delta V(t)-\epsilon_1\sum_{l=1}^n(e^l(t)^2)^{\frac{1+p}{2}}-\epsilon_2\sum_{l=1}^n(e^l(t)^2)^{\frac{1+q}{2}}\\
\le&2\delta V(t)-\epsilon_1\bigg(\sum_{l=1}^ne^l(t)^2\bigg)^{\frac{1+p}{2}}-\epsilon_2n^{\frac{1-q}{2}}\bigg(\sum_{l=1}^ne^l(t)^2\bigg)^{\frac{1+q}{2}}\\
=&2\delta V(t)-\epsilon_12^{\frac{1+p}{2}}V(t)^{\frac{1+p}{2}}-\epsilon_2n^{\frac{1-q}{2}}2^{\frac{1+q}{2}}V(t)^{\frac{1+q}{2}}\\
\le &\left\{\begin{array}{ll}
-(\overline{\alpha}-2\delta)V(t)^{\frac{1+p}{2}};&~~V(t)<1\\
-(\overline{\beta}-2\delta) V(t)^{\frac{1+q}{2}};&~~V(t)\ge 1
\end{array}\right.
\end{align*}
According to Lemma \ref{t2}, one can get that the fixed-time synchronization is finally realized, and the settling time can be also obtained as (\ref{settling8}).
\end{proof}

Next, we will apply the above theorem on the fixed-time stabilization of equilibrium point for nonlinear systems, including neural networks.

Consider the nonlinear function described by:
\begin{align}\label{a1}
\dot{x}(t)=W_1x(t)+W_2\Phi(x(t))+J
\end{align}
where $x(t)\in R^n$ is the state vector, $\Phi(x): R^n\rightarrow R^n$ is the nonlinear function satisfying $\|\Phi(x)-\Phi(y)\|\le \|W_3(x-y)\|, \forall x,y\in R^n$, and $J$ is the external disturbance vector. If we need to control the system (\ref{a1}) to the desired state $x^{\star}$, where $x^{\star}$ is the equilibrium point of (\ref{a1}), $ W_1x^{\star}+W_2\Phi(x^{\star})+J=0.$

\begin{corollary}{\label{apply}}
For the nonlinear system (\ref{a1}), if there exist two scalars $\delta>0$ and $\varepsilon>0$, such that
\begin{align}\label{LMI}
\frac{1}{2}({W_1+W_1^T}+\varepsilon W_2W_2^T+\varepsilon^{-1}W_3^TW_3)\le \delta I,
\end{align}
then under the following control
\begin{align}\label{NN}
\dot{x}(t)=&W_1x(t)+W_2\Phi(x(t))+J\nonumber\\
&-\epsilon_1\mathrm{sig}^p(x(t)-x^{\star})-\epsilon_2\mathrm{sig}^q(x(t)-x^{\star}),
\end{align}
the fixed-time stability can be realized, and the settling time is defined by (\ref{settling8}).
\end{corollary}
\begin{proof}
Denote $f(x)=W_1x(t)+W_2\phi(x(t))+J$, one just needs to prove that the QUAD condition (\ref{QU}) holds. Denote $e(t)=x(t)-x^{\star}$, and $\tilde{\Phi}(e(t))=\Phi(x(t))-\Phi(x^{\star})$, one can get
\begin{align*}
&(x(t)-x^{\star})^T(f(x(t))-f(x^{\star}))\\
=&(x(t)-x^{\star})^T[W_1(x(t)-x^{\star})+W_2(\Phi(x(t))-\Phi(x^{\star}))]\\
=&e(t)^T[W_1e(t)+W_2\tilde{\Phi}(e(t))]\\
\le&e(t)^TW_1e(t)+\varepsilon e(t)^TW_2W_2^Te(t)+\varepsilon^{-1}\tilde{\Phi}(e(t))^T\tilde{\Phi}(e(t))\\
\le&\frac{1}{2}e(t)^T\bigg[{W_1+W_1^T}+\varepsilon W_2W_2^T+\varepsilon^{-1}W_3^TW_3\bigg]e(t)\\
\le &\delta e(t)^Te(t).
\end{align*}
Then, according to Theorem \ref{master-slave}, we can get the conclusion.
\end{proof}

\begin{remark}
As for the existence of $\delta$ in condition (\ref{LMI}), one can first convert it to the matrix form, then use the Matlab toolbox Linear Matrix Inequality (LMI) to solve it.

On the other hand, one can also estimate the term as: $e(t)^TW_2\tilde{\Phi}(e(t))\le \|W_2\|\|W_3\|e(t)^Te(t)$, in this case,
\begin{align}\label{neural}
\delta\ge\lambda_{\max}(\frac{W_1+W_1^T}{2})+\|W_2\|\|W_3\|.
\end{align}
\end{remark}

\begin{remark}
In \cite{LHYC14}, the stabilization of nonlinear systems is also investigated, but they concern on finite-time stabilization, while in this paper we investigate the fixed-time stabilization. Of course, one can also use $V(t)=\frac{1}{2}\sum_{i=1}^n\xi_ie_i(t)^2, \xi_i>0, i=1,\cdots,n$ to investigate the fixed-time stabilization of the above Corollary, but in order to keep in accordance with Theorem \ref{master-slave}, here we don't consider the parameters $\xi_i$.
\end{remark}

\section{Numerical examples}\label{sim}
In this section, a simple numerical example is given to demonstrate the correctness of obtained theoretical results.

Consider a network of five agents, and the original dynamical behavior $x(t)$ of each node is described by is a 3-D neural network satisfying:
\begin{align}\label{simu2}
\dot{x}(t)=f(x(t))=-x(t)+W\Phi(x(t)),
\end{align}
where $x(t)=(x^1(t),x^2(t),x^3(t))^T$,
$\Phi(x(t))=(\phi(x^1(t)),\\\phi(x^2(t)),\phi(x^3(t)))^T$, $\phi(v)=(|v+1|-|v-1|)/2$,
and
\begin{align*}
W=\left(\begin{array}{rrr} 1.25&-3.2&-3.2\\-3.2&1.1&-4.4\\-3.2&4.4&1
\end{array}\right).
\end{align*}
This neural network has a double-scrolling chaotic attractor, see
Fig. \ref{chaotic}.
\begin{figure}
\begin{center}
\includegraphics[width=0.5\textwidth,height=0.25\textheight]{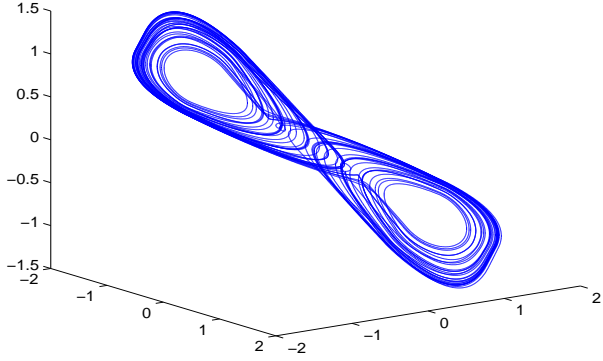}
\end{center}
\caption{Chaotic attractor of 3-D neural network
(\ref{simu2}) with initial values $(0.4,0.1,-0.2)^T$.} \label{chaotic}
\end{figure}
For this function $f(\cdot)$, from (\ref{neural}), one can easily get that $\delta=6.1$ can make the QUAD condition (\ref{QU}) hold.

Suppose the network is partitioned into two clusters as $\mathcal{C}_1=\{1,2\}$ and $\mathcal{C}_2=\{3,4,5\}$. As for the coupling matrix, we assume that
\begin{align*}
A=B=\left(\begin{array}{rr}
A_{11}&A_{12}\\
A_{21}&A_{22}
\end{array}\right),
\end{align*}
where
\begin{align*}
&A_{11}=\left(\begin{array}{rr}
-1&1\\
1&-1
\end{array}\right),
A_{12}=\left(\begin{array}{rrr}
-0.1&0.3&-0.2\\
0.1&-0.3&0.2
\end{array}\right),\\
&A_{21}=A_{12}^T, ~~~~~~~~~~~~~~~~~~~
A_{22}=\left(\begin{array}{rrrr}
-2&1&1\\
1&-2&1\\
1&1&-2
\end{array}\right).
\end{align*}

Assume the pinning controllers are added on the first node of each cluster, i.e., only node $1$ and $3$ are pinned. In this case, the network can be described by the following equations:
\begin{align}
&\dot{x}_{1}(t)=f(x_{1}(t))\nonumber\\
&+\alpha \cdot a_{12}\mathrm{sig}^p\bigg(x_2(t)-x_{1}(t)\bigg)+\mathrm{sig}^p\bigg(\sum_{j=3}^5a_{1j}x_j(t)\bigg)\nonumber\\
&+\beta\cdot b_{12}\mathrm{sig}^q\bigg(x_2(t)-x_{1}(t)\bigg)
+\mathrm{sig}^q\bigg(\sum_{j=3}^5b_{1j}x_j(t)\bigg)\nonumber\\
&-\alpha\mathrm{sig}^p(x_{1}(t)-s_1(t))-\beta\mathrm{sig}^q(x_{1}(t)-s_1(t));\nonumber\\
&\dot{x}_{2}(t)=f(x_{2}(t))\nonumber\\
&+\alpha \cdot a_{21}\mathrm{sig}^p\bigg(x_1(t)-x_{2}(t)\bigg)+\mathrm{sig}^p\bigg(\sum_{j=3}^5a_{2j}x_j(t)\bigg)\nonumber\\
&+\beta\cdot b_{21}\mathrm{sig}^q\bigg(x_1(t)-x_{2}(t)\bigg)
+\mathrm{sig}^q\bigg(\sum_{j=3}^5b_{2j}x_j(t)\bigg)\nonumber\\
&\dot{x}_{3}(t)=f(x_{3}(t))\nonumber\\
&+\alpha \sum_{j=4,5}a_{3j}\mathrm{sig}^p\bigg(x_j(t)-x_{3}(t)\bigg)+\mathrm{sig}^p\bigg(\sum_{j=1,2}a_{3j}x_j(t)\bigg)\nonumber\\
&+\beta\sum_{j=4,5}b_{3j}\mathrm{sig}^q
\bigg(x_j(t)-x_{3}(t)\bigg)+\mathrm{sig}^q\bigg(\sum_{j=1,2}b_{3j}x_j(t)\bigg)\nonumber\\
&-\alpha\mathrm{sig}^p(x_{3}(t)-s_2(t))-\beta\mathrm{sig}^q(x_{1}(t)-s_2(t));\nonumber\\
&\dot{x}_{i}(t)=f(x_{i}(t))\nonumber\\
&+\alpha \sum_{j\in\mathcal{C}_2,j\ne i}a_{ij}\mathrm{sig}^p\bigg(x_j(t)-x_{i}(t)\bigg)+\mathrm{sig}^p\bigg(\sum_{j=1,2}a_{ij}x_j(t)\bigg)\nonumber\\
&+\beta\sum_{j\in\mathcal{C}_2,j\ne i}b_{ij}\mathrm{sig}^q
\bigg(x_j(t)-x_{i}(t)\bigg)+\mathrm{sig}^q\bigg(\sum_{j=1,2}b_{ij}x_j(t)\bigg)\nonumber\\
&~~~~~~~~~~~~~~~~~~~~~~~~~~~~~~~~~~~~~~~~~~~~~~~~~~~i=4,5\label{simulation}
\end{align}
where $s_1(t)$ and $s_2(t)$ are two different trajectories described by (\ref{simu2}) with initial values $s_1(0)=(0.4,0.1,-0.2)^T$ and $s_2(0)=(0.1,0.1,0.1)^T$, respectively.

Now, for $p=0.5$ and $q=2$, according to notations in (\ref{pg1}) and (\ref{pg2}), one can get that
\begin{align*}
\hat{A}_{11}=-2A_{11}+\mathrm{diag}\{2^{\frac{4}{3}},0\}=
\left(\begin{array}{rr}
4.5198&2\\
2&-2
\end{array}\right),\\
\hat{A}_{22}=-2A_{22}+\mathrm{diag}\{2^{\frac{4}{3}},0,0\}
=\left(\begin{array}{rrrr}
6.5198&-2&-2\\
-2&4&-2\\
-2&-2&4
\end{array}\right),\\
\hat{B}_{11}=-2B_{11}+\mathrm{diag}\{2^{\frac{2}{3}},0\}=
\left(\begin{array}{rr}
3.5874&2\\
2&-2
\end{array}\right),\\
\hat{A}_{22}=-2A_{22}+\mathrm{diag}\{2^{\frac{2}{3}},0,0\}=\left(\begin{array}{rrrr}
5.5874&-2&-2\\
-2&4&-2\\
-2&-2&4
\end{array}\right),
\end{align*}
so parameters in (\ref{rho})-(\ref{g2}) are: $\rho_1=0.6395$, $\rho_2=0.4445$, $\overline{N}=18$, and
\begin{align}\label{parameter}
\overline{\alpha}=0.6013\alpha, \overline{\beta}=0.0988\beta, \gamma_1=5.4385, \gamma_2=0.5091.
\end{align}
According to Theorem \ref{main}, one can get that if $\overline{\alpha}-\gamma_1-2\delta>0$ and $\overline{\beta}-\gamma_2-2\delta>0$, then $\alpha\ge 29.3339$ and $\beta>128.5617$.

Define an index
\begin{align}
E(t)=\sqrt{\sum_{i=1}^2\|x_i-s_1(t)\|^2+\sum_{i=3}^5\|x_i-s_2(t)\|^2}
\end{align}
for cluster synchronization error. Let $\alpha=30$ and $\beta=130$, then one can get the settling time defined by (\ref{settling}) is $7.3956$, while the real settling time is about $0.1735$, which is far less than the theoretical value, see Fig. \ref{fixedsyn}.
\begin{figure}
\begin{center}
\includegraphics[width=0.5\textwidth,height=0.25\textheight]{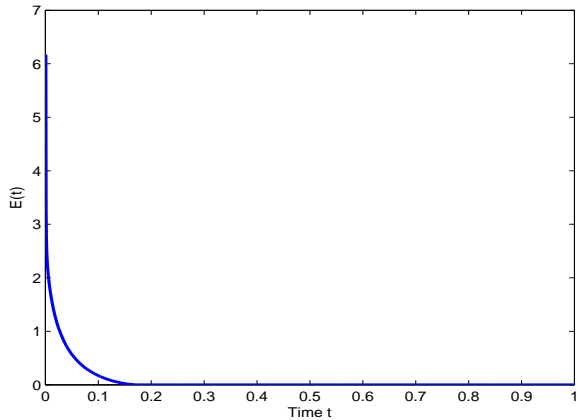}
\end{center}
\caption{Fixed-time cluster synchronization for network (\ref{simulation}) with the settling time $T_{\max}=0.1375$ for $\alpha=30$ and $\beta=130$.} \label{fixedsyn}
\end{figure}

In fact, when parameters $\alpha$ and $\beta$ are less than the calculated theoretical values, the fixed-time cluster synchronization can also be realized. In the following, we discuss the fixed-time cluster synchronization under different values of parameters, like $\alpha, \beta, p, q$, see Fig. \ref{diff-a}-Fig. \ref{diff-q}. All these simulations can coincide with our previous analysis about fixed-time cluster synchronization.

\begin{figure}
\begin{center}
\includegraphics[width=0.5\textwidth,height=0.25\textheight]{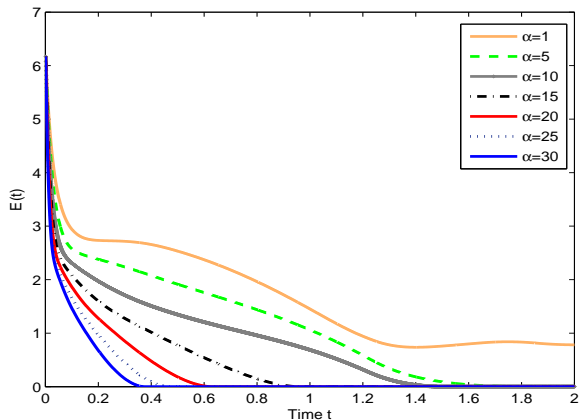}
\end{center}
\caption{Fixed-time cluster synchronization for network (\ref{simulation}) under $\beta=1$ and $\alpha=1,5,10,15,20,25,30$ respectively. Obviously, the larger the $\alpha$ is, the faster the cluster synchronization can be realized.} \label{diff-a}
\end{figure}

\begin{figure}
\begin{center}
\includegraphics[width=0.48\textwidth,height=0.25\textheight]{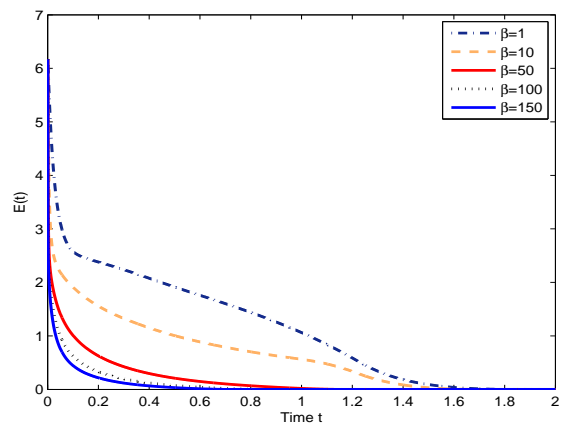}
\end{center}
\caption{Fixed-time cluster synchronization for network (\ref{simulation}) under $\alpha=5$ and $\beta=1,10,50,100,150$ respectively. Obviously, the larger the $\beta$ is, the faster the cluster synchronization can be realized.} \label{diff-b}
\end{figure}

\begin{figure}
\begin{center}
\includegraphics[width=0.5\textwidth,height=0.25\textheight]{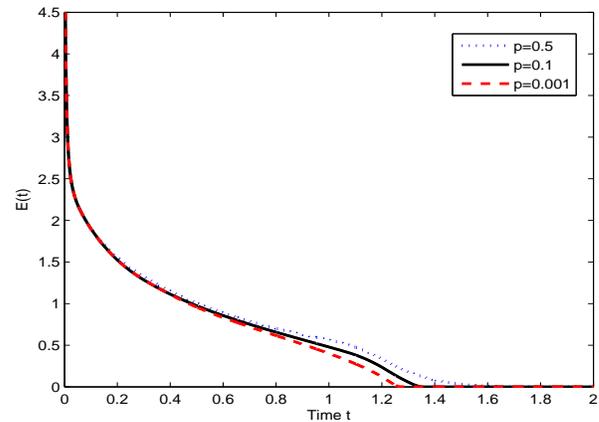}
\end{center}
\caption{Fixed-time cluster synchronization for network (\ref{simulation}) with $\alpha=5$, $\beta=10$, $q=2$ and $p=0.5, 0.1, 0.001$ respectively. One can find that, the smaller the $p$ is, the faster the cluster synchronization can be realized. Moreover, the times used for the network error from the initial values to $1$ are almost the same, since the parameter $q$ in these cases are the same.} \label{diff-p}
\end{figure}

\begin{figure}
\begin{center}
\includegraphics[width=0.48\textwidth,height=0.25\textheight]{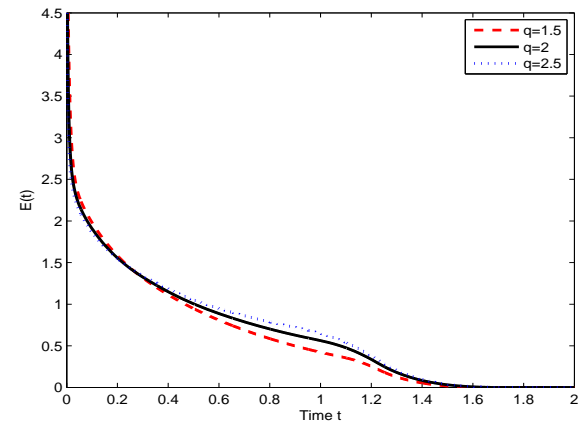}
\end{center}
\caption{Fixed-time cluster synchronization for network (\ref{simulation}) with $\alpha=5$, $\beta=10$, $p=0.5$ and $q=1.5, 2, 2.5$ respectively. One can find that, the settling times are almost the same, since the parameter $p$ in these cases are the same.} \label{diff-q}
\end{figure}

\section{Conclusion}\label{con}
In this paper, we study the fixed-time cluster synchronization problem, while previous works often concentrate on fixed-time consensus problems without control or finite-time complete synchronization/consensus problems. We first design a new distributed protocol which can be used to realize the fixed-time cluster synchronization, then rigorous proofs are given to show the validity of this new protocol. Moreover, when all the nodes in the network lie in the same cluster, it becomes the fixed-time complete synchronization problem, which is also carefully discussed, because it can contain the master-slave coupled case and the stability of the nonlinear systems case, which is applied on the investigation of the fixed-time stabilization of the equilibrium in neural networks. Finally, some numerical simulations are presented to show the correctness of our obtained theoretical results.

It is hoped that this paper may shed some light on the study
of fixed-time cluster synchronization via pinning control.
However, there are still many challenging problems to be investigated in the next step, for example: 1) the unified form for general fixed-time cluster synchronization and finding its key property; 2) the fixed-time cluster synchronization for directed networks, which is more general in the real world, while in this paper only undirected network topology is discussed; 3) the fixed-time cluster synchronization for networks with time delays; 4) the fixed-time cluster synchronization for networks with adaptive coupling mechanisms.


\end{document}